\documentclass[11pt,english]{article}

\usepackage[margin=1in]{geometry}
\usepackage{libertine}

\usepackage{amsmath,amsthm,mathtools,amssymb}
\usepackage{thmtools,thm-restate}
\usepackage{algorithm}
\usepackage[noend]{algorithmic}
\usepackage[shortlabels]{enumitem}
\usepackage[usenames,svgnames,xcdraw,table]{xcolor}
\definecolor{DarkBlue}{rgb}{0.1,0.1,0.5}
\definecolor{DarkGreen}{rgb}{0.1,0.5,0.1}
\usepackage{hyperref}
\hypersetup{
   colorlinks   = true,
 linkcolor    = DarkBlue, 
 urlcolor     = DarkBlue, 
	 citecolor    = DarkGreen 
}

\usepackage[capitalize]{cleveref}
\usepackage{graphicx}
\usepackage{svg}
\usepackage{float}
\usepackage{verbatim}
\usepackage{caption}

\usepackage{pgfplots}
\usepackage{pifont}

\newcounter{casenum}

\DeclareMathOperator*{\argmax}{arg\,max}
\DeclareMathOperator*{\argmin}{arg\,min}
\DeclareMathOperator*{\integers}{\mathbb{Z}}
\DeclareMathOperator*{\lo}{LO}
\DeclareMathOperator*{\dist}{dist}
\DeclareMathOperator*{\res}{Res}

\renewcommand{\algorithmiccomment}[1]{\bgroup\hfill\small{$\blacktriangleright$ #1}\egroup}

\newcommand{\eqx}{\textrm{EQx}}

\newcommand{\efx}{\textrm{EFX}}
\newcommand{\ef}{\textrm{EF}}

\newcommand{\struct}[1]{\ensuremath{\textrm{St}_{#1}}}%

\newcommand{\cac}[2]{\textsc{Cut}(#1,#2)}

\newtheorem{theorem}{Theorem}

\newtheorem{claim}[theorem]{Claim}
\newtheorem{lemma}[theorem]{Lemma}
\newtheorem{proposition}[theorem]{Proposition}
\newtheorem{definition}{Definition}

\title{\bfseries EFX Allocations on Some Multi-graph Classes}

\author{
Umang Bhaskar\thanks{Tata Institute of Fundamental Research. {\tt umang@tifr.res.in}} \quad Yeshwant Pandit\thanks{Tata Institute of Fundamental Research. {\tt yeshwant.pandit@tifr.res.in}} 
}

\date{}

\begin{document}

\maketitle




\begin{abstract}
    The existence of EFX allocations is one of the most significant open questions in fair division. Recent work by Christodolou, Fiat, Koutsoupias, and Sgouritsa (``Fair allocation in graphs'', EC 2023) establishes the existence of EFX allocations for graphical valuations, when agents are vertices in a graph, items are edges, and each item has zero value for all agents other than those at its end-points. Thus in this setting, each good has non-zero value for at most two agents, and there is at most one good valued by any pair of agents. This marks one of the few cases when an exact and complete EFX allocation is known to exist for arbitrary agents. 
    
    In this work, we extend these results to multi-graphs, when each pair of vertices can have more than one edge between them. The existence of EFX allocations in multi-graphs is a natural open question given their existence in simple graphs. We show that EFX allocations exist, and can be computed in polynomial time, for agents with cancellable valuations in the following cases: (i) bipartite multi-graphs, (ii) multi-trees with monotone valuations, and (iii) multi-graphs with girth $(2t-1)$, where $t$ is the chromatic number of the multi-graph. The existence in multi-cycles follows from (i) and (iii).
\end{abstract}

\section{Introduction}

We study the existence and computation of allocations that are envy-free up to any item, or EFX, in instances that can be represented as multi-graphs. The question of whether EFX allocations exist in general instances is possibly the most tantalising open question in the fair allocation of resources, and many recent papers have focused on this question, leading to proofs of its existence in ever-larger classes of instances.

A fair division instance consists of a set of items that is to be allocated to a set of agents, with each item allocated to an agent. Each agent has a monotone nondecreasing value function over the set of items. Given an allocation, agent $i$ \emph{envies} agent $j$ if agent $i$'s value for its own allocation is less than its value for the set of items allocated to $j$. An envy-free allocation is one where no agent envies another. Envy-free allocations do not exist in general, and the closest relaxation then studied is envy-freeness upto any good, or EFX. An EFX allocation allows an agent to envy another, but the removal of \emph{any} item from the envied agent's allocation should eliminate the envy as well. 

EFX allocations were first studied by Caragiannis et al.~\cite{CKM+19unreasonable}, who established its relation to other notions of fairness, but left the question of existence open. If all agents have the same valuation function, then an EFX allocation has been shown to exist~\cite{PlautR20}. This can be extended to show the existence for two agents with distinct valuations. This was extended to multiple agents with two distinct additive valuations~\cite{Mahara23}, then to three agents when at least one has an MMS-feasible valuation~\cite{AkramiACGMM23,ChaudhuryGM24}, and recently to multiple agents with three distinct additive valuations~\cite{PrakashGNV24}. A number of papers also study partial EFX allocations, when some items may remain unallocated, e.g.,~\cite{ChaudhuryKMS21,BergerCFF22}, as well as approximately EFX allocations~\cite{AmanatidisFS24,AmanatidisMN20,Chan0LW19,PlautR20}.

Given the difficulties in establishing the existence of EFX allocations with multiple agents, researchers study special cases of structured valuations. One such class of structured valuations is \emph{graphical valuations} when the instance can be represented as a graph~\cite{ChristodolouFKS23}. Here the agents correspond to the vertices of the graph, and each item is an edge $e = \{u,v\}$. The restriction is that each item has nonzero marginal value for only the agents at its end-points, and has zero value for all other agents. Thus if $e = \{i,j\}$, then for every other agent $k$ and subset $S$ of items, $v_k(S \cup \{e\}) = v_k(S)$. Further, the graph is simple, and thus there is at most one good valued by every pair of agents. For this case, and agents with monotone valuations over incident edges, Christodolou et al.~\cite{ChristodolouFKS23} show that EFX allocations exist, marking one of the rare cases where EFX allocations are shown to exist with multiple agents.

A natural question with graphical valuations is if an EFX orientation exists--- an EFX allocation where goods are only allocated to agents at their endpoints. Christodolou et al. show that an EFX orientation may not exist, and it is NP-hard to determine if a given instance has an EFX orientation~\cite{ChristodolouFKS23}. Zeng and Mehta characterize graphs that permit EFX orientations~\cite{ZengM24}.

Given the existence of EFX in simple graphs, a natural question is if these results can be extended to multi-graphs when there are multiple edges between vertices (equivalently, pairs of agents can value multiple goods). Here a $2/3$-approximate EFX was shown to exist~\cite{AmanatidisFS24}, beating the earlier bound of $0.618$ for general additive valuations. In independent and concurrent work, Afshinmehr et al.~\cite{AfshinmehrDKMR24} show that exact EFX allocations exist in bipartite multi-graphs and multi-cycles. They also characterize when EFX orientations exist in bipartite graphs, and show it is NP-hard to decide if a given instance admits an EFX orientation.

\subsection{Our contribution}

We study EFX allocations in multi-graphs, and show that EFX allocations exist in large and important subclasses of multi-graphs. Our results are algorithmic, and for agents with cancellable valuations (defined in the next section), we show an EFX allocation can be obtained in polynomial time. A feature of our algorithms is that they are iterative and fairly simple to state, and are based on the well-known cut-and-choose paradigm. While the algorithms are natural, the analysis is complicated and somewhat subtle.

\begin{enumerate}
    \item Firstly, we show for bipartite multi-graphs with cancellable valuations, an EFX allocation can be obtained in polynomial time.
    \item Secondly, we show that for multi-trees, i.e., when there are no cycles of length three or larger, EFX allocations exist for general monotone valuations. Since multi-trees are also bipartite graphs, for multi-trees with cancellable valuations, a polynomial time algorithm immediately follows.
    \item Lastly, we generalize the result for bipartite multi-graphs, and show that for multi-graphs with girth at least $(2t-1)$, where $t$ is the chromatic number of the graph, EFX allocations can be obtained in polynomial time for agents with cancellable valuations. Our algorithm here is a natural extension of the algorithm for bipartite multi-graphs, though the analysis is more complicated.
\end{enumerate}

We note that any multi-cycle is 3-colourable. Hence our last result shows the existence of EFX allocations in multi-cycles of length 5 or greater. A 2-colourable multi-cycle is bipartite, for which existence follows from the first result. The existence in multi-cycles of length 3 is known, and hence our work also shows the existence of EFX allocations for multi-cycles.

For the last result, we assume that we are given a $t$-colouring of the graph. A concrete example of such a graph is the Petersen graph, with chromatic number 3 and girth 5. Graphs of large chromatic number and girth are very interesting mathematically and are often used as counterexamples for various conjectures. A number of papers study the construction of such graphs. For example, Erd\H{o}s shows that for any integers $k$ and $l$, there exist graphs with chromatic number at least $k$ and girth at least $l$ using the probabilistic method (e.g.,~\cite{AlonS16}). Explicit constructions 
 of such graphs are also known, including through the use of Ramanujan graphs.    

\subsection{Comparison to work by Afshinmehr et al.~\cite{AfshinmehrDKMR24}}
\label{sec:comparison}

As mentioned, the problem of obtaining EFX allocations in bipartite multi-graphs is also studied independently and concurrently by Afshinmehr et al.~\cite{AfshinmehrDKMR24}, who show that EFX allocations exist in bipartite multi-graphs and multi-cycles. It is instructive to compare our algorithm for bipartite multi-graphs with theirs. Their algorithm is stated as a two-phase procedure similar to earlier work~\cite{ChristodolouFKS23}, where in the first phase an EFX orientation with certain additional properties is found, and then the EFX orientation is converted into a complete allocation through the use of `safe' vertices. Stating the requirements of the first phase for EFX orientations requires a number of additional definitions, including the notion of available bundles. However, it turns out that the allocations obtained through both algorithms are similar, though not identical. E.g., in Section 4.4, Case 2, the algorithm by Afshinmehr et al. allocates edges differently from our algorithm. For bipartite multi-graphs, this does not make a difference, but extending this to multi-graphs could possibly complicate the analysis. 

We believe our algorithm and analysis for bipartite multi-graphs to be much simpler and clearer. This allows us to extend the algorithm in two different directions. In Section~\ref{sec:tree}, we show that the algorithm can be modified to show the existence of EFX allocations in tree multi-graphs with arbitrary monotone valuations. Secondly, in Section~\ref{sec:chromatic}, we extend the algorithm to show the existence of EFX allocations in multi-graphs with girth at least $2t-1$, with chromatic number $t$. The algorithm for this latter result is in fact a fairly straightforward extension of our algorithm for biparitite multi-graphs. However, the analysis is more complicated and requires keeping track of the bundles as they move around the multi-graph. The algorithm as stated in Afshinmehr et al. does not immediately lend itself to any such natural extension beyond bipartite multi-graphs. 

Lastly, as stated, Afshinmehr et al. also show the existence of EFX allocations in multi-cycles. We obtain the same result as an implication of our last result, for $t$-chromatic multi-graphs.

\section{Notation and Preliminaries}

\paragraph{multi-graphs.} A multi-graph $G = (V,E)$ is a graph with possibly multiple edges between each pair of vertices. A function $r$ maps each edge $e \in E$ to an unordered pair of vertices, which are its endpoints. We use $E_{u,w}$ for the set of parallel edges between the vertices $u$, $w$ (hence $E_{u,w} = E_{w,u}$). We use $n := |V|$ and $m := |E|$. For a vertex $u$, $N_u$ is the set of neighbours $\{w \in V: \, \exists e \in E, \, r(e) = \{u,w\}\}.$ 

Given a multi-graph, we can define paths and cycles as in simple graphs without parallel edges. A path $P= (v_1, \ldots, v_k)$ is a sequence of distinct vertices so that each pair of consecutive vertices are neighbours. The length of a path is one less than the number of vertices. A cycle $C = (v_1, \ldots, v_{k-1}, v_k)$ is a sequence of distinct vertices so that each pair of consecutive vertices are neighbours, and $v_1$, $v_k$ are neighbours. 

\paragraph{Fair division.} In a fair division instance on a multi-graph $G=(V,E)$, each vertex $u \in V$ corresponds to an agent $u$, and each edge $e$ corresponds to a good. Throughout this article, we will use agents and vertices interchangeably, and edges and goods interchangeably. Each agent $u$ has a valuation function $v_u: 2^E \rightarrow \integers_+$ that maps subsets of goods to an integer value. A \emph{bundle} $S \subseteq E$ is simply a subset of goods.

The multi-graph restricts the valuation functions in the following way: Each agent $u \in V$ has zero marginal value for any good not adjacent to it. Thus for a good $g$ if $u \not \in r(g)$, then for any subset of goods $S \subseteq E$, $v_u(S \cup \{g\}) = v(S)$. We say that agent $u$ \emph{values} a good $g$ if $u \in r(g)$, i.e., the edge $g$ is adjacent to $u$, else it does not value the good $g$. Extending this definition, we say that agent $u$ values a bundle $S \subseteq E$ if $S$ contains some good that $u$ values, else $u$ does not value the bundle $S$.

An allocation $A = (A_u)_{u \in V}$ is a partition of the set of items (or edges) among the agents. A partial allocation is a partition of a subset of items. An allocation $A$ is \emph{envy-free} (EF) if for all agents $u$, $w$, $v_u(A_u) \geq v_u(A_w)$.  An allocation $A$ is \emph{envy-free up to any good} (EFX) if for all agents $u$, $w$, either $v_u(A_u) \ge v_u(A_w)$, or $v_u(A_u) \geq v_u(A_w \setminus \{x\})$, for all $x \in A_w$.

Given an allocation $A$ (or a partial allocation), we can define an \emph{envy graph} $G_A$ among the agents: The vertices of $G_A$ are the agents, and there is a directed edge $(u,w)$ if agent $u$ envies agent $w$, i.e., $v_u(A_u) < v_u(A_w)$. An orientation is an allocation where each item is allocated to an agent that values it (thus if $r(e) = \{u,w\}$, then $e \in A_u \cup A_w$).

Given an allocation $A$ and the corresponding envy graph $G_A$, let $C = (u_1, u_2, \ldots, u_k)$ be a directed cycle in the envy graph. We use the term \emph{resolving the envy cycle $C$} to refer to the modification that exchanges bundles along the cycle: agent $u_k$ gets the bundle currently allocated to $u_1$, and each agent $u_i$ gets the bundle currently allocated to $u_{i+1}$, for $i < k$. The modified allocation is called $A^C$. 

\paragraph*{Valuation functions.}

We will be interested in a few different classes of valuation functions. A valuation function $v:2^E \rightarrow \integers_+$ is \emph{monotone nondecreasing} (or just monotone) if for every $S \subseteq E$ and $g \in E$, $v(S \cup \{g\}) \ge v(S)$.

A monotone valuation function is \emph{cancellable} if for $S$, $T \subseteq E$ and $g \not \in S \cup T$, if $v(S) \ge v(T)$ then $v(S \cup \{g\}) \ge v(T \cup \{g\})$. These functions were defined in prior work on EFX allocations for four agents~\cite{BergerCFF21,BergerCFF22}. We only consider monotone cancellable functions, and hence say cancellable to mean monotone cancellable functions. We in particular need the following two properties of cancellable functions.

\begin{proposition}
Let $v$ be a cancellable function. Then given sets of goods $S_1$, $S_2$, $T_1$, $T_2$ so that set $S_1$ and $S_2$ are disjoint, sets $T_1$, $T_2$ are disjoint, and $v(S_1) \ge v(T_1)$, $v(S_2) \ge v(T_2)$, it follows that $v(S_1 \cup S_2) \ge v(T_1 \cup T_2)$. 
    \label{prop:cancellable-1}
\end{proposition}

For two agents, EFX allocations are known to exist for monotone valuations~\cite{PlautR20}. It is known however that even for two agents with submodular valuations, computing an EFX allocation is PLS-complete and requires an exponential number of queries~\cite{GoldbergHH23}. However, for cancellable valuations, we can compute an EFX allocation in polynomial time~\cite{GoldbergHH23}.\footnote{The actual result holds for the more general class of \emph{weakly well-layered} valuations.} 

\begin{proposition}
\label{prop:cancellable-2agent-efx}
For two identical agents with cancellable valuations, an EFX allocation can be computed in polynomial time.
\end{proposition}

\paragraph*{Cut and Choose.}

Cut and choose is a standard procedure to obtain an EFX allocation for two agents. Given agents $u$ and $w$ and a set of items $S$, one agent `cuts' and the other agent `chooses.' If agent $u$ `cuts', we assume there are two identical agents with valuation function $v_u$, and use the algorithm from~\Cref{prop:cancellable-2agent-efx} to obtain a partition $(S^1,S^2)$ of $S$ that is EFX for identical agents with valuation function $v_u$. Then agent $w$ `chooses' the bundle with a higher value for it, while agent $u$ is assigned the other bundle. The allocation is clearly envy-free for agent $w$, and is EFX for agent $u$. Note that for agents with cancellable valuations, this gives us a polynomial time algorithm to obtain an EFX allocation.

The cut-and-choose procedure is central to our algorithm. In~\Cref{sec:bipartite}, when we deal with bipartite multi-graphs $G=(L \cup R, E)$, for allocating a set of parallel edges $E_{u,w}$ with $u \in L$, $w \in R$, agent $w \in R$ will always cut (thus the agent in the right bipartition always cuts). In~\Cref{sec:chromatic}, when we deal with $t$-chromatic multi-graphs, and partition the vertex set $V = C_1 \cup \ldots \cup C_t$ where vertices in $C_i$ are the same colour, we think of the vertices in $C_i$ as being to the left of $C_{i+1}$. For a set of parallel edges $E_{u,w}$, as in the bipartite case, the vertex to the \emph{right} will always cut.

We use $\cac{u}{S}$ to refer to the procedure where we use the algorithm from~\Cref{prop:cancellable-2agent-efx} to obtain a partition $(S^1,S^2)$ of $S$ that is EFX for identical agents with valuation function $v_u$.

\section{EFX Allocations in Bipartite Multi-graphs}
\label{sec:bipartite}

We first show that for bipartite multi-graphs and agents with cancellable valuations, an EFX allocation always exists, and can be computed in polynomial time. Let $V = L \cup R$ be the bipartition of the vertex set. We first define some notation we will use for this case.

\paragraph*{Structures}

For each vertex $u \in L$, we define the \emph{structure} $\struct{u}$ as the subgraph induced by $u \cup N_u$. Note that since $N_u \subseteq R$, all edges in $\struct{u}$ are between $u$ and vertices in $N_u$. We say $u$ is the \emph{root} of the structure. Each iteration in our algorithm will \emph{resolve} the structure $\struct{u}$, for some $u \in L$, by which we mean that it will assign all the edges in $\struct{u}$. We also say that a vertex $u$ is resolved to mean that $\struct{u}$ is resolved.

For each $w \in N_u$, let $(E_{u,w}^1, E_{u,w}^2) = \cac{w}{E_{u,w}}$ be the partition of $E_{u,w}$ returned when agent $w$ cuts. Define $S_{u,w} := \arg \max \{v_u(E_{u,w}^1), v_u(E_{u,w}^2)\}$ as the bundle preferred by $u \in L$, and $T_{u,w} := \arg \max \{v_w(E_{u,w}^1), v_w(E_{u,w}^2)\}$ as the bundle preferred by $w$. If either agent is indifferent between the two bundles, we break ties so that $S_{u,w} \neq T_{u,w}$. Further, for $S \subseteq E_{u,w}$, we define $\bar{S} = E_{u,w} \setminus S$.

For a structure $\struct{u}$, we further define a \emph{favourite} neighbour $f_u$ as follows:

\[
f_u := \arg \max_{w \in N_u} \max \{v_u(S_{u,w})\} \, .
\]

\noindent Then $f_u$ is the neighbour that offers $u$ the highest-value bundle after cutting the adjacent edges. The other neighbours are simply called \emph{ordinary} neighbours.

Our algorithm is then simple to describe, and is formally given as~\Cref{alg:BipartiteEFX}. We consider the vertices in $L$ in turn. For vertex $u \in L$, we consider the structure $\struct{u}$ and resolve it, as follows. Each neighbour $w \in N_u$ cuts the bundle $E_{u,w}$. Each \emph{ordinary} neighbour $w \neq f_u$ gets their higher-value bundle $T_{u,w}$. Let $\lo_u = \bigcup_{w \in N_u \setminus \{f_u\}} \bar{T}_{u,w}$ be the union of the remaining goods from the ordinary neighbours of $u$. For the favourite neighbour $f_u$, if $S_{u,f_u} \neq T_{u,f_u}$ --- both $u$ and $f_u$ prefer different bundles when $f_u$ cuts $E_{u,f_u}$ --- then $u$ gets $S_{u,f_u} \cup \lo_u$, and $f_u$ gets $T_{u,f_u}$. In this case, everyone gets their largest value bundle, and there is no envy. On the other hand, if $S_{u,f_u} = T_{u,f_u}$ --- both $u$ and $f_u$ prefer the same bundle when $f_u$ cuts $E_{u,f_u}$ --- then $u$ is offered the choice of $S_{u,f_u}$ or $\bar{S}_{u,f_u} \cup \lo_u$. If $u$ prefers the former, then it gets $S_{u,f_u}$, and $f_u$ gets $\bar{S}_{u,f_u} \cup R_u$. In this case, $f_u$ may envy $u$, but this is EFX. If $u$ prefers  $\bar{S}_{u,f_u} \cup \lo_u$, it gets this bundle, and $f_u$ gets $S_{f_u,u} (=T_{u,f_u})$. Again, in this case, there is no envy.

\begin{algorithm}[ht]
\caption{Bipartite-EFX}
\label{alg:BipartiteEFX}
\begin{algorithmic}[1] 
\REQUIRE Bipartite multi-graph $G = (L \cup R, E)$, and vertices with cancellable valuations over $E$.
\ENSURE \efx{} allocation $A$.
\STATE Initially, for each vertex $u$, $A_u=\emptyset$ \label{line:b-init}
\FOR{each vertex $u \in L$} \label{line:b-outerforloop1}
     \FOR{each vertex $w \in N_u$} \label{line:b-innerloop1}
          \STATE $(E_{u,w}^1, E_{u,w}^2)= \cac{w}{E_{u,w}}$. \COMMENT{Returns an EFX partition for identical agents}\label{line:b-cutloop}
          \STATE $S_{u,w}=\argmax_{S \in \{E^{1}_{u,w},E^{2}_{u,w}\}} v_u(S)$
          \STATE $T_{u,w}=\argmax_{S \in \{E^{1}_{u,w},E^{2}_{u,w}\}} v_w(S)$. \COMMENT{Preferred bundles for $u$ and $w$ from $E_{u,w}$}
     \ENDFOR  
     \STATE $f_u=\argmax_{w \in N_u} v_u(S_{u,w})$ \label{line:b-favourite-neighbour}
     \STATE $\lo_u = \bigcup_{w \in N_u \setminus \{f_u\}} \bar{T}_{u,w}$ \COMMENT{Left-over bundles from all ordinary neigbours}
     \STATE $A_w = A_w \cup T_{u,w}$ for each vertex $w \in N_u \setminus \{f_u\}$ \COMMENT{Ordinary neighbours get their preferred bundle} \label{line:b-ordinary-neighbours}
     \IF[Both $u$ and $f_u$ prefer the same bundle in $E_{u,f_u}$]{($S_{u,f_u} = T_{u,f_u}$)}
            \IF[$u$ prefers $S_{u,f_u}$ over all the left-over bundles]{($v_u(S_{u,f_u}) > v_u(\bar{S}_{u,f_u} \cup \lo_u)$)} \label{line:b-condn1-S}
                \STATE $A_u = S_{u,f_u}$
                \STATE $A_{f_u} = A_{f_u} \cup \bar{S}_{u,f_u} \cup \lo_u$
            \ELSE[$u$ prefers the left-over bundles] \label{line:b-condn2-leftovers}
                \STATE $A_u = \bar{S}_{u,f_u} \cup \lo_u$
                \STATE $A_{f_u} = A_{f_u} \cup  S_{u,f_u}$
            \ENDIF
    \ELSE[$u$ and $f_u$ prefer different bundles] \label{line:b-condn3-allok}
        \STATE $A_u = S_{u,f_u} \cup \lo_u$
        \STATE $A_{f_u} = A_{f_u} \cup T_{u,f_u}$
    \ENDIF
\ENDFOR
\RETURN $A$
\end{algorithmic}
\end{algorithm} 

We note the following properties of the algorithm. In each iteration of the outer for loop, a vertex $u \in L$ is chosen, and the goods in $\struct{u}$ are assigned to the agents in $\struct{u}$. The allocation to all other agents remains unchanged. Since $\struct{u}$ does not contain any agent from $L$ other than $u$, an agent in $u' \in L$ has no allocated goods until it is resolved, and the allocation to agent $u'$ is not changed after it is resolved. Finally, no goods are removed from an agent once assigned.

\begin{theorem}
    \Cref{alg:BipartiteEFX} runs in polynomial time, and returns an EFX allocation.
    \label{thm:b-envy}
\end{theorem}

The polynomial running time is easily seen, and is because the procedure $\cac{u}{S}$ runs in polynomial time for cancellable valuations. The EFX property follows immediately from the next lemma.

\begin{lemma}
    \label{lem:b-envy}
    Fix an agent $u \in L$. Let $A$ be the allocation after resolving $\struct{u}$. Then the only possible envy in allocation $A$ is that a resolved agent $u' \in L$ is envied by its favourite agent $f_{u'}$. Further, the partial allocation $A$ is EFX.
\end{lemma}

\begin{proof}
    The proof is by induction on the iterations of the outer for loop. The claim is clearly true for the initial empty allocation. Let $\hat{A}$ be the allocation before $u$ is resolved. As noted, any vertex $u' \in L$, $u' \neq u$ cannot be in $\struct{u}$. Then by the induction hypothesis, in $\hat{A}$, no agents in $\struct{u}$ are envied.
    
    In the current iteration, when $\struct{u}$ is resolved, only agents in $\struct{u}$ get goods, while agents not in $\struct{u}$ retain their allocation. Hence, any new envy edges must be to agents in $\struct{u}$. Further since only agents in $\struct{u}$ value the goods in $\struct{u}$, any new envy edges must also be from agents in $\struct{u}$. 

    Consider any pre-existing envy from an agent $w$ to $w' \not \in \struct{u}$ (as established, agents in $\struct{u}$ are not envied in $\hat{A}$). Thus $v_w(\hat{A}_w) < v_w(\hat{A}_{w'})$. In the current iteration, the allocation to $w'$ does not change, while $w$ does not lose any goods. Hence any such envy remains EFX. Thus to prove the lemma, we only need to show that any new envy between agents in $\struct{u}$ is from $f_u$ to $u$, and is EFX.

    Consider first an ordinary neighbour $w$. Of the goods allocated in this iteration, $w$ is only interested in $E_{u,w}$. Of this set, it gets its preferred partition $T_{u,w}$. Then since $w$ did not envy anyone in $\struct{u}$ earlier, for any agent $z \in \struct{u}$, $v_w(\hat{A}_w) \ge v_w(\hat{A}_z)$. Then
    
    \[
    v_w(A_w) = v_w(\hat{A}_w \cup T_{u,w}) \ge v_w(\hat{A}_z \cup \bar{T}_{u,w}) \ge v_w(A_z) \, .
    \]

    \noindent The first inequality is because valuations are cancellable, and $w$ prefers $\hat{A}_w$ to $\hat{A}_z$, and $T_{u,w}$ to $\bar{T}_{u,w}$. The second inequality is because in this iteration, any agent other than $w$ either does not receive any good valued by $w$, or receives the bundle $\bar{T}_{u,w}$.

    Consider $f_u$, agent $u$'s favourite neighbour. Again of the goods allocated in this iteration, $f_u$ is only interested in $E_{u,f_u}$. Of this set, it either gets its preferred bundle $T_{u,f_u}$ (in either Line~\ref{line:b-condn2-leftovers} or in Line~\ref{line:b-condn3-allok}), or it gets the bundle $\bar{S}_{u,f_u}$ in Line~\ref{line:b-condn1-S}. In the former case, since the valuations are cancellable and agent $f_u$ did not earlier envy any agent in $\struct{u}$, agent $f_u$ does not envy any agent in $\struct{u}$ in allocation $A$ either. In the latter case, $f_u$ gets the bundle $\bar{S}_{u,f_u}$, and hence may envy $u$ since agent $u$ gets $S_{u,f_u} = T_{u,f_u}$. But note that then $A_{u} = S_{u,f_u}$, and this is an EFX partition since $f_u$ cuts the set of items $E_{u,f_u}$. Thus, the only possible new envy edge from $f_u$ is to $u$, and this is EFX.

    Lastly, consider agent $u$. We show that $u$ will not envy any agent in $\struct{u}$. If $u$ gets $S_{u, f_u} \cup \lo_u$ (Line~\ref{line:b-condn3-allok}), then by definition $S_{u,f_u}$ is preferred over all the other bundles allocated to the other agents in $\struct{u}$, and $u$ additionally gets all the left-over bundles. Otherwise, $u$ gets a choice between $S_{u,f_u}$ and $\bar{S}_{u,f_u} \cup \lo_u$; the other bundle is given to $f_u$. Then clearly $u$ does not envy $f_u$. Both bundles offered have higher value for $u$ than $T_{u,w}$ for any neighbour $w$, and hence $u$ does not envy any ordinary neighbour either.
\end{proof}

\section{EFX Allocations for Monotone Valuations in Tree Multi-graphs}
\label{sec:tree}

We now show that for monotone valuations, an EFX allocation exists in tree multi-graphs. Note that even for two agents with monotone submodular valuations, computing an EFX allocation is PLS-complete~\cite{GoldbergHH23}, and hence a polynomial-time algorithm for tree multi-graphs may not exist. Since trees are also bipartite, a polynomial time algorithm for cancellable valuations follows from the previous section.

Our algorithm for trees is recursive, and again utilizes the procedure $\cac{v}{S}$, that obtains an EFX partition $(S^1, S^2)$ for two identical agents with valuation $v$. Let $T=(V, E)$ be the tree multi-graph, and $\ell$ be a leaf with parent $p$. Then $E_{\ell,p}$ are the edges between $\ell$ and $p$, and no other agent values these goods. Inductively, let $A'$ be an EFX allocation of goods $E \setminus E_{\ell,p}$ to agents $V \setminus \{\ell\}$. Our objective is then to extend $A'$ to an EFX allocation of goods $E$.

For this, let $(S^1, S^2)$ be the partition obtained when $p$ cuts the bundle $E_{\ell,p}$. Agent $\ell$ gets its preferred bundle (say $S^1$). If $p$ is not envied in allocation $A'$, then it gets $S^2$, and the algorithm terminates. If $p$ is already envied in allocation $A'$, it cannot get additional goods. Let $s$ then be a source in the envy graph with a path to $p$. We add $S^2$ to agent $s$'s allocation. If $p$ does not envy $s$ now, the algorithm terminates. Else, there is an envy cycle where $p$ envies $s$. Resolving the envy cycle gives us an EFX allocation.

\begin{algorithm}[H]
\caption{Tree-EFX}
\label{alg:TreeEFX}
\begin{algorithmic}[1] 
\REQUIRE A tree multi-graph $T = (V, E)$ with $|V| \ge 2$ and monotone valuations.
\ENSURE An \efx{} allocation $A$
\IF{($|V|=2$)}
    \STATE Let $V = \{\ell,p\}$
    \STATE $(S^1,S^2) = \cac{p}{E_{\ell,p}}$
    \STATE $A_\ell = \arg \max_{S \in \{S^1, S^2\}} v_\ell(S)$, $A_p = \bar{A}_\ell$ 
    \RETURN $A = (A_\ell, A_p)$
\ENDIF
\STATE Let $\ell$ be a leaf in $T$
\STATE $A' = \textrm{Tree-EFX}(V \setminus \{\ell\},E \setminus E_{\ell,p})$ \COMMENT{Recursive call on the smaller tree without leaf $\ell$}
\STATE $G_{A'}$ is the envy graph for allocation $A'$
\WHILE[{Resolve any envy-cycles}]{(there is an envy-cycle $C$ in $G_{A'}$)}
    \STATE $A' = {A'}^C$
\ENDWHILE
\STATE $A_u  = A'_u$ for all vertices $u \neq \ell$ \label{line:t-recursive}
\STATE $(S^1, S^2) = \cac{p}{E_{\ell,p}}$  \COMMENT{$p$ cuts the bundle $E_{\ell,p}$}
\STATE $A_\ell = \arg \max_{S \in \{S^1, S^2\}} v_\ell(S)$ \COMMENT{Leaf $\ell$ gets its preferred bundle} \label{line:t-ell}
\IF{($p$ is a source in $G_{A'}$)} \label{line:t-easy}
    \STATE $A_p = A_p \cup \bar{A}_\ell$ \COMMENT{If $p$ is not envied, $p$ gets the left-over items} 
\ELSE
    \STATE Let $s$ be a source with a path $P$ to $p$. \COMMENT{Else, assign left-over items to $s$, resolve any envy cycle.}
    \STATE $A_s = A_s \cup \bar{A}_\ell$
    \IF{($v_p(A_p) < v_p(A_s)$)}
        \STATE Let $C$ be the envy-cycle consisting of the edge $(p,s)$ and the $s$-$p$ path $P$.
        \STATE $A = A^C$
    \ENDIF
\ENDIF
\RETURN $A$
\end{algorithmic}
\end{algorithm} 

\begin{theorem}
    Given a tree multi-graph $T = (V,E)$ with monotone valuations, Algorithm \textrm{Tree-EFX} returns an EFX allocation $A$.
    \label{thm:tree-efx}
\end{theorem}

\begin{proof}
    By induction, assume that $A'$ is an EFX allocation on $T' = (V \setminus \{\ell\},E \setminus E_{\ell,p})$. Note that given an EFX allocation, resolving envy cycles does not destroy this property, since each agent's value does not decrease, and bundles are merely reassigned. Hence, the partial allocation in Line~\ref{line:t-recursive} is EFX for agents $V \setminus \{\ell\}$, and $A_{\ell} = \emptyset$.

    On allocating $A_\ell$ to agent $\ell$, clearly $\ell$ does not envy any agent. The only agent that could envy $\ell$ is agent $p$, since only agents $p$ and $\ell$ value the bundle $A_\ell$. Further, since $p$ gets to cut the bundle $E_{\ell,p}$, as long as $p$ possesses a bundle with value at least $v_p(\bar{A}_\ell)$, its envy towards $\ell$ will be EFX.

    If $p$ is a source in the envy graph, then the condition in Line~\ref{line:t-easy} is satisfied and gets the bundle $\bar{A}_\ell$. No agent will envy $p$, since agents $V \setminus \{p,\ell\}$ do not value the bundle $\bar{A}_\ell$, and $\ell$ already gets its preferred bundle. Thus the allocation $A$ is EFX.

    If $p$ is not a source, then since the envy-graph is acyclic (and $\ell$ does not envy any agent), there is a source $s$ with a path to $p$ in the envy graph. Let $P$ be this path. We assign $\bar{A}_\ell$ to agent $s$. The only agent that could now envy $s$ is $p$. If $p$ envies $s$, there is an envy cycle, which we resolve. We claim that the resulting allocation is EFX. This is because agent $\ell$ continues to not envy any agent. Agent $p$'s value for its bundle is now at least $v_p(\bar{A}_\ell)$, and hence any envy towards agent $\ell$ is EFX. Further agent $p$'s value for its bundle is also at least its earlier value; since resolving the cycle only shifted bundles around, any envy towards agents $V \setminus \{\ell\}$ is also EFX by induction. Finally, for other agents, their value does not decrease, and any envy is unaffected by the shifting of bundles. Hence the allocation is EFX.
\end{proof}

\section{EFX Allocations in $t$-Chromatic multi-graphs with Girth at Least $(2t-1)$}
\label{sec:chromatic}

We now extend our algorithm and analysis to show that EFX allocations exist for agents with cancellable valuations in multi-graphs with girth at least $2t-1$, where $t$ is the chromatic number of the graph. A multi-graph is $t$-chromatic if $t$ is the minimum number of colours required to colour the vertices, so no edge has both end-points the same colour. The girth of a multi-graph is the length of the shortest simple cycle, without using parallel edges. For example, bipartite multi-graphs are $2$-chromatic and have girth at least $4$, satisfying the assumption.\footnote{Clearly, any $2$-chromatic multi-graph has girth at least $4$.} Another example is the Petersen graph (with parallel edges), which is $3$-colourable and has girth $5$. We will assume that we are given a $t$-colouring of the multi-graph.

Any cycle is 3-colourable, hence our algorithm gives an EFX allocation for cycles of length at least 5. A cycle of length 4 is a bipartite graph, hence~\Cref{alg:BipartiteEFX} gives an EFX allocation for this case. For 3 agents, EFX allocations are known to exist. Hence, our work shows that in cycles with parallel edges and cancellable valuations, an EFX allocation can be computed in polynomial time.

Our algorithm proceeds as follows. Given a multi-graph $G=(V,E)$ coloured with $t$ colours, let $c(u) \in [t]$ be the colour of each vertex. Let $C^i = \{u: c(u) = i\}$ be the independent set consisting of all vertices coloured $i$. We think of the vertices as ordered from left to right by their colour, and say a vertex $u$ is \emph{to the left of} $w$ if $c(u) < c(w)$. For $i = 1, \ldots, t-1$, define $L^i := C^i$, and $R^i := C^{i+1} \cup \ldots \cup C^t$.

We proceed in phases. In phase $i$, we consider the bipartite multi-graph $G^i = (V^i = L^i \cup R^i, E^i)$ where an edge $e \in E$ is in $E^i$ if $|r(e) \cap L^i|$ $= |r(e) \cap R^i| = 1$, i.e., each edge is from a vertex in $L^i$ to a vertex in $R^i$. Note that each vertex $u \in V$ appears in a left bipartition in exactly one phase, when we consider all edges to vertices to the right of $u$. Thus each edge appears in exactly one phase.

Now in phase $i$, we run~\Cref{alg:BipartiteEFX} on the bipartite multi-graph $G^i$. The only change we make is that in~\Cref{alg:BipartiteEFX}, initially all vertices --- in particular, vertices $u \in L$ --- had empty allocations. This may not be the case in the current algorithm in the second phase onwards. Hence, we remove Line~\ref{line:b-init}, which initializes the allocation. Further in Lines~\ref{line:c-condn1-S} and~\ref{line:c-condn3-allok}, agent $u \in L^i$ retains its earlier allocation. In Line~\ref{line:c-condn2-leftovers}, when $u \in L^i$ is possibly envied by $f_u$, its earlier allocation passes to $f_u$.~\Cref{alg:chromaticEFX} formally states the algorithm.

\begin{algorithm}[ht]
\caption{$t$-Chromatic-EFX}
\label{alg:chromaticEFX}
\begin{algorithmic}[1] 
\REQUIRE $t$-Chromatic multi-graph $G = (L \cup R, E)$ with girth $(2t-1)$ and cancellable valuations. $C^i$ is the set of vertices coloured $i \in [t]$.
\ENSURE \efx{} allocation $A$.
\STATE Initially, for each vertex $u$, $A_u=\emptyset$. Resolved agents $\res = \emptyset$. \label{line:c-init}
\FOR{$i =1$ to $t-1$}
\STATE $L^i = C^i$, $R^i = C^{i+1} \sqcup \ldots \sqcup C^t$. \COMMENT{Only consider edges between $L^i$, $R^i$}
\FOR{each vertex $u \in L^i$} \label{line:c-outerforloop1}
     \FOR{each vertex $w \in N_u \cap R^i$} \label{line:c-innerloop1}
          \STATE $(E_{u,w}^1, E_{u,w}^2)= \cac{w}{E_{u,w}}$. \COMMENT{Returns an EFX partition for identical agents}\label{line:c-cutloop}
          \STATE $S_{u,w}=\argmax_{S \in \{E^{1}_{u,w},E^{2}_{u,w}\}} v_u(S)$
          \STATE $T_{u,w}=\argmax_{S \in \{E^{1}_{u,w},E^{2}_{u,w}\}} v_w(S)$. \COMMENT{Preferred bundles for $u$ and $w$ from $E_{u,w}$}
     \ENDFOR  
     \STATE $f_u=\argmax_{w \in N_u \cap R^i} v_u(S_{u,w})$ \label{line:c-favourite-neighbour}
     \STATE $\lo_u = \bigcup_{w \in N_u \setminus \{f_u\} \cap R^i} \bar{T}_{u,w}$ \COMMENT{Left-over bundles from all ordinary neigbours}
     \STATE $A_w = A_w \cup T_{u,w}$ for all $w \in N_u \setminus \{f_u\} \cap R^i$ \COMMENT{Ordinary neighbours get their preferred bundle} \label{line:c-ordinary-neighbours}
     \IF[Both $u$ and $f_u$ prefer the same bundle in $E_{u,f_u}$]{($S_{u,f_u} = T_{u,f_u}$)}
            \IF[$u$ prefers $S_{u,f_u}$ over left-over and its own goods]{($v_u(S_{u,f_u}) > v_u(A_u \cup \bar{S}_{u,f_u} \cup \lo_u)$)} \label{line:c-condn1-S}
                \STATE $A_{f_u} = A_{f_u} \cup A_u \cup \bar{S}_{u,f_u} \cup \lo_u$
                \STATE $A_u = S_{u,f_u}$
            \ELSE[$u$ prefers the left-over bundles] \label{line:c-condn2-leftovers}
                \STATE $A_u = A_u \cup \bar{S}_{u,f_u} \cup \lo_u$
                \STATE $A_{f_u} = A_{f_u} \cup  S_{u,f_u}$
            \ENDIF
    \ELSE[$u$ and $f_u$ prefer different bundles] \label{line:c-condn3-allok}
        \STATE $A_u = A_u \cup S_{u,f_u} \cup \lo_u$
        \STATE $A_{f_u} = A_{f_u} \cup T_{u,f_u}$
    \ENDIF
    \STATE $\res = \res \cup \{u\}$ \COMMENT{Bookkeeping: $u$ is resolved.}
    \ENDFOR
\ENDFOR
\RETURN $A$
\end{algorithmic}
\end{algorithm} 

A \emph{phase} is an iteration of the outer loop. In phase $i$, the vertices in $C^i$ are picked sequentially in the inner for loop.

\paragraph*{Structures.} For our analysis, we slightly need to redefine \emph{structures}. For a vertex $u \in C^i$, we now define the structure $\struct{u}$ as the subgraph induced by $u \cup (N_u \cap R^i)$. That is, a structure now consists of $u$ and all neighbours \emph{to its right}. As before, $u$ is the \emph{root} of the structure. Now each iteration in phase $i$ of our algorithm resolves the structure $\struct{u}$ for some $u \in C^i$, and thus assigns all the edges in $\struct{u}$. A vertex $u$ is resolved if $\struct{u}$ is resolved.

Fix agent $u \in C^i$. As before, for each $w \in (N_u \cap R^i)$, $(E_{u,w}^1, E_{u,w}^2) = \cac{w}{E_{u,w}}$ is the partition of $E_{u,w}$ returned when agent $w$ cuts. Define $S_{u,w} := \arg \max \{v_u(E_{u,w}^1), v_u(E_{u,w}^2)\}$ as the bundle preferred by $u \in L$, and $T_{u,w} := \arg \max \{v_w(E_{u,w}^1), v_w(E_{u,w}^2)\}$ as the bundle preferred by $w$. If either agent is indifferent between the two bundles, we break ties so that $S_{u,w} \neq T_{u,w}$. Further, for $S \subseteq E_{u,w}$, we define $\bar{S} = E_{u,w} \setminus S$. Then $u$'s \emph{favourite} neighbour $f_u$ is $f_u := \arg \max_{w \in N_u \cap R_i} \max \{v_u(S_{u,w})\}$.

\noindent This is the right neighbour that offers $u$ the highest-value bundle after cutting the adjacent edges. The other right neighbours are simply called \emph{ordinary} neighbours. 

We note the following properties of our algorithm. In phase $i$ of the algorithm, the inner for loop picks an agent $u \in C^i$, and resolves $\struct{u}$. The goods in $\struct{u}$ are assigned to the agents in $\struct{u}$. The allocation to other agents is unchanged. Hence in particular, the allocation to a resolved agent does not change subsequently. While $\struct{u}$ is resolved, except for $u$, agents in $\struct{u}$ do not lose any goods. Agent $u$ may however give its entire set of goods to its favourite neighbour $f_u$ (Line~\ref{line:c-condn1-S}). However, in this case, agent $u$ gets $S_{u,f_u}$, which it prefers over its earlier allocation and all the left-over goods. Thus, the value of each agent for its bundle is nondecreasing.


Much of the effort in our analysis goes to proving~\Cref{lem:c-newenvy}, that any envy is within a structure, and is from a favourite neighbour $f_u$ to the root $u$. In particular, the issue we want to avoid is the following. Say for some interim allocation $A$, agent $u'$ values goods in both $A_u$ and $A_{f_u}$. Then if $A_u$ is transferred to $f_u$ in Line~\ref{line:c-condn1-S}, agent $u'$ should not start envying $f_u$.

We show this by contradiction: that if agent $u'$ does indeed start envying $f_u$, then it must have short paths to both $u$ and $f_u$, and hence, taken with any edge between $u$ and $f_u$, there is a cycle of length less than $2t-1$. We do this in two steps:~\Cref{clm:c-dist-splgoods} shows that if $u'$ envies the bundle of goods $A_u \cup A_{f_u}$, then these must contain a good from $\struct{u'}$. Claims~\ref{clm:c-dist-easy} and~\ref{clm:c-dist-splgoods} then establish bounds on the distance from $u'$ for $u$ and $f_u$.

We first show that an unresolved agent does not envy the \emph{entire} set of goods held by all other unresolved vertices.

\begin{claim}
Let $A$ and $\res$ be the allocation and the set of resolved vertices at Line~\ref{line:c-innerloop1}. Let $z$ be an unresolved agent. Then 

\[v_z(A_z) \ge v_z\left(\bigcup_{w \not \in \res, \, w \neq z} A_w\right) \, .\]
\label{clm:c-envy-union1}
\end{claim}

\begin{proof}
    Fix an unresolved vertex $z$. We prove this by induction on the iterations of the inner loop. Initially, the allocation is empty, and the claim is trivially satisfied. Assume this is true prior to the loop when a vertex $u$ is resolved. Let $\hat{A}$, $\widehat{\res}$ be the allocation and set of resolved vertices before $u$ is resolved, and $A$, $\res = \widehat{\res} \cup \{u\}$ be the allocation and set of resolved vertices after. Note that the allocation to resolved agents does not change subsequently. Hence if $z$ now envies the set of goods held by unresolved agents, new goods that $z$ values must be allocated to unresolved agents.

    In this loop, only the allocation to agents in $\struct{u}$ is modified. Further, none of the agents in $\struct{u}$ are resolved prior to the loop, i.e., $\struct{u} \cap \widehat{\res} = \emptyset$. Hence $v_z(\hat{A}_z) \ge v_z\left(\cup_{u' \in \struct{u}, \, u' \neq z} \hat{A}_{u'}\right)$ by the induction hypothesis. Now if $z \not \in \struct{u}$, then no additional goods it values are allocated in this loop. Hence the claim holds by induction. Thus, since $z\neq u$ (since $z$ is unresolved after the loop), it must be that $z$ is a neighbour to the right of $u$. Of the goods allocated in this loop, $z$ only values $E_{u,z}$.

    If $z$ is an ordinary neighbour, then $z$ additionally gets its preferred bundle $T_{u,z}$, while the bundle $\bar{T}_{u,z}$ is allocated to other agents. Hence, since valuations are cancellable, the claim remains true. On the other hand, if $z = f_u$, then $z$ could lose its preferred bundle $T_{u,z}$. But in this case agent $u$ is resolved, hence $u \not \in \res$. Since $z$ additionally gets $\bar{T}_{u,z}$ and does not lose any goods, the claim remains true after the loop also.
\end{proof}

The next claim builds on the previous claim to show that if any vertex $z$ envies the union of goods held by some set $S$ of unresolved agents, then $z$ must be resolved, and some agent in $S$ must hold a good from $\struct{z}$.

\begin{claim}
Let $A$ and $\res$ be the allocation and the set of resolved vertices at Line~\ref{line:c-innerloop1}. Let $z$ be any agent, and $S$ is a subset of unresolved vertices so that

\[
v_z(A_z) < v_z \left(\bigcup_{w \in S, z \not \in S} A_w \right) \, .
\]

Then $z$ is resolved, and for some agent $w \in S$, $A_w$ contains some good $g  \in \struct{z}$.
\label{clm:c-envy-union2}
\end{claim}

\begin{proof}
By Claim~\ref{clm:c-envy-union1}, if $z$ is unresolved, then $v_z(A_z)$ $\ge v_z \left( \cup_{w \in S, \, w \neq z} A_w \right)$. Hence $z$ must in fact be resolved. We prove the contrapositive by induction, that if a subset of unresolved agents $S$ do not hold goods from $\struct{z}$, then $v_z(A_z) \ge v_z(\cup_{w \in S} A_w)$. Clearly, the claim holds in any iteration before $z$ is resolved.

Consider the iteration where $z$ is resolved. Every agent in $\struct{z}$ receives goods from $\struct{z}$, hence, we only consider subsets $S$ consisting of agents not in $\struct{z}$. But the allocation for these agents does not change, and $z$'s value for its bundle is nondecreasing, hence the claim holds in this iteration.

Consider any later iteration, where an agent $u$ is resolved. Then $u \not \in S$, since $u$ is resolved. Let $\hat{A}$ be the allocation before resolving $\struct{u}$, and $A$ the allocation after. Again, only the allocation for agents in $\struct{u}$ changes, hence by induction we only need to consider subsets $S$ that contain agents in $\struct{u}$. Note that since agent $z$ is already resolved, $z$ does not value any goods in $\struct{u}$. Hence the allocation of goods in $\struct{u}$ does not affect the right hand side of the inequality in the claim.

Other than the allocation of goods in $\struct{u}$, the only change that happens in the algorithm is that the prior bundle $\hat{A}_u$ is possibly transferred to $f_u$. Thus if $\hat{A}_u$ is not transferred to $f_u$, or $f_u \not \in S$, then $z$'s value for the aggregate bundle $\cup_{w \in S} A_w$ does not change, and the claim holds.

Thus, $f_u \in S$, $u \not \in S$, and the only relevant change that occurs in the iteration is that the bundle $\hat{A}_u$ is transferred to agent $f_u$. Note that since $f_u \in S$, $f_u$ does not hold any goods from $\struct{z}$, and hence $\hat{A}_u$ does not contain any goods from $\struct{z}$. Now in the allocation $\hat{A}$, consider the subset $\hat{S} = S \cup \{u\}$. Then

\begin{align*}
v_z(A_z) = v_z(\hat{A}_z) & \ge v_z\left(\bigcup_{w \in \hat{S}} \hat{A}_w\right) \qquad \text{(by the induction hypothesis)}\\
    & = v_z\left(\hat{A}_u \cup \bigcup_{w \in S} \hat{A}_w \right) \qquad \text{(since $\hat{S} = S \cup \{u\}$)}\\
    & = v_z\left(\bigcup_{w \in S} A_w \right) \, ,
\end{align*} 
proving the claim. 
\end{proof}

The next few results show that a good that is valued by an agent $z$ cannot be very far from $z$. The next proposition simply shows that a good held by an agent $z$ can only be moved when $\struct{z}$ is resolved, to agent $f_z$.

\begin{proposition}
For an agent $u$, let $\hat{A}$ and $A$ be the allocations before and after $\struct{u}$ is resolved. If for some agent $z$ and good $g$, $g \in \hat{A}_z$ but $g \not \in A_z$, then $z = u$ and $g \in A_{f_u}$. Thus in a phase, any good allocated before the phase began can only be transferred from the root of a structure to its favourite neighbour.
\label{prop:c-moveone}
\end{proposition}

\begin{proof}
As noted, when $\struct{u}$ is resolved, the goods in the structure are allocated to the agents in the structure. If Line~\ref{line:c-condn1-S} is executed, then additionally the bundle $A_u$ is transferred to agent $f_u$ (and agent $u$ gets the bundle $S_{u,f_u})$. Thus the only way that a good that is assigned to an agent before $\struct{u}$ is resolved can be moved, is if it was assigned to $u$ and then is transferred to agent $f_u$. The proposition follows. 

Since a root only appears once in a phase, a good once transferred in a phase cannot be transferred again, and hence any good allocated prior to a phase can only be transferred once, from the root of a structure to its favourite neighbour. 
\end{proof}

\begin{claim}
Let $A$ be the allocation at Line~\ref{line:c-innerloop1}. If agent $z$ values a good $g$ and $g \in A_w$, then $\dist(z,w) \le c(w)$ along allocated edges.
\label{clm:c-dist-easy}
\end{claim}

\begin{proof}
By~\Cref{prop:c-moveone}, a good can only be transferred rightward, from an agent being resolved to its favourite neighbour. If $c(w) = 1$, then once $\struct{w}$ is resolved, $w$'s allocation does not change. Thus if $g \in A_w$ and $c(w)=1$, $w$ received this good in phase 1 when $\struct{w}$ was resolved. Since $z$ values $g$, $z \in N_w$, and there is path of length 1 between $z$ and $w$ along the edge set $E_{w,z}$. 

Now suppose $c(w) \ge 2$. Consider the phase when $g$ is initially allocated, say to an agent $w'$ when $\struct{u}$ is resolved. Since $z$ values $g$, $z \in \struct{u}$. Hence $z$ and $w'$ have a path of length 2 along the edges of $\struct{u}$. From~\Cref{prop:c-moveone}, in each subsequent phase, $g$ is transferred from the root of a structure to its favourite neighbour. Suppose $g$ is allocated to $w$ in phase $k$. Then $w$ must be a right neighbour for the agent being resolved, and hence the colour $c(w) \ge k+1$. There are at most $k-1$ phases after the phase $g$ is initially allocated. Thus there is a path of length $k-1+2 \le c(w)$ along allocated edges from $w$ to $z$.
\end{proof}

\begin{claim}
Let $A$ be the allocation at Line~\ref{line:c-innerloop1}. Let there be agents $u$, $w$, and good $g \in \struct{u} \cap A_w$. Then $\dist(u,w) \le c(w)-c(u)$ along allocated edges.
\label{clm:c-dist-splgoods}
\end{claim}

\begin{proof}
Since $g \in \struct{u}$, when $\struct{u}$ is resolved, it is assigned either to $u$ or to a neighbour to the right. Further any reassignment of $g$ transfers it to a neighbour further to the right. Hence if $u \neq w$, then $w$ must be to the right of $u$, hence $c(w) > c(u)$.

Now when $g$ is initially allocated in phase $c(u)$, it is allocated to an agent $w'$ that is at a distance 1 from $u$ along the edges $E_{u,w'}$. Suppose $g$ is allocated to agent $w$ in phase $k$. Then $c(w) > k$, and there are at most $k-c(u)$ phases after the phase when $g$ is initially allocated. Thus there is a path of length $k-c(u)+1 \le c(w)-c(u)$ along allocated edges from $u$ to $w$.
\end{proof}


\begin{lemma}
    Let $A$ be the allocation at Line~\ref{line:c-innerloop1}. Then if $w$ envies $u$, then $u$ is a resolved vertex, and $w = f_u$ (i.e., $w$ is $u$'s favourite vertex in $\struct{u}$).
    \label{lem:c-newenvy}
\end{lemma}

\begin{proof}
    We will prove the lemma by induction. For the base case, the initial allocation is empty and trivially satisfies the two properties. Now fix any vertex $u$. Let $\hat{A}$ be the allocation before the structure rooted at $u$ is resolved, and $A$ be the allocation after resolving $\struct{u}$.
    
    None of these vertices in $\struct{u}$ are resolved in $\hat{A}$, and hence by the induction hypothesis, none of these are envied in $\hat{A}$. When $\struct{u}$ is resolved, only the allocation to agents in $\struct{u}$ is modified, while the allocation to other agents is unchanged. The value of agents in $\struct{u}$ does not decrease. Hence, any new envy edge must be towards agents in $\struct{u}$. 

    Consider an ordinary neighbour $w$ in $\struct{u}$. Agent $w$ is not envied in $\hat{A}$. It receives its preferred bundle $T_{u,w}$; this is only valued by $u$ and $w$, hence no agent other than $u$ will envy $w$. To see that $u$ also will not envy $w$, we first claim that $v_u(\hat{A}_w) = 0$, i.e., $u$ does not value $w$'s bundle in $\hat{A}$. To prove this, assume for a contradiction that $u$ values $w$'s bundle. Then by~\Cref{clm:c-dist-easy}, there exists a $u$-$w$ path along allocated edges of length at most $t$. Including the (unallocated) edges $E_{u,w}$, this gives a cycle of length $t+1$. If $t=2$, this gives us a 3-cycle, which would mean that the graph is not 2-colourable. If $t>2$, then $t+1 < 2t-1$, but we assume that the graph has girth at least $2t-1$. In either case we get a contradiction. Thus, $v_u(\hat{A}_w) = 0$.

    Further, $v_u(A_u) \ge v_u(T_{u,w})$. This is because agent $u$ prefers $S_{u,f_u}$ over $T_{u,w}$, and always has the choice of keeping $S_{u,f_u}$. Hence, it ends up with a bundle of value at least that of $T_{u,w}$. Then since $v_u(\emptyset) \ge v_u(\hat{A}_w)$ and $v_u(A_u) \ge v_u(T_{u,w})$, since valuations are cancellable, $v_u(A_u) \ge v_u(A_w)$. Thus, an ordinary neighbour $w \in \struct{u}$ is not envied.

    Now consider agent $u$. As shown, agent $u$ is not envied in $\hat{u}$. While resolving $\struct{u}$, the only goods agent $u$ can possibly receive are goods in $\struct{u}$, and hence no agent outside $\struct{u}$ will envy $u$. Within $\struct{u}$, all ordinary  neighbours receive their preferred bundle, and hence will not envy $u$. Thus the only agent that could possibly envy $u$ after resolving $\struct{u}$ is $f_u$.

    Now if $f_u$ gets its preferred bundle $T_{u,f_u}$ (in Lines~\ref{line:c-condn2-leftovers} and~\ref{line:c-condn3-allok}), it does not envy $u$. If agent $f_u$ gets $\bar{T}_{u,f_u}$ (in Line~\ref{line:c-condn1-S}), then this is an EFX partition for $f_u$ (since $f_u$ cut s$E_{u,f_u}$), and agent $u$ only gets $T_{u,f_u}$. In this case, $f_u$ envies $u$, but the envy is EFX.

    Lastly, consider agent $f_u$. This case is more complicated than the earlier cases, since the bundle $\hat{A}_u$ previously held by $u$ could be transferred to $f_u$ in Line~\ref{line:c-condn1-S}. We will show however that no agent envies $f_u$.

     First, let's consider agents in $\struct{u}$. Again, the ordinary neighbours $w$ did not envy $f_u$ earlier, and from the goods $E_{u,w}$ they get their preferred bundle $T_{u,w}$. Hence, since valuations are cancellable, they will not envy $f_u$. 

     Agent $u$, as argued previously, by~\Cref{clm:c-dist-easy}, does not value $\hat{A}_{f_u}$, else there is either a 3-cycle or a cycle of length strictly smaller than $2t-1$. Further agent $u$ gets a bundle of value at least that of $S_{u,f_u}$, and hence will not envy $f_u$.

     For the last case, consider any agent $u' \not \in \struct{u}$. In allocation $\hat{A}$, $u'$ does not envy $u$ or $f_u$. Since $u'$ does not value goods in $E_{u,f_u}$, if it envies $A_{f_u}$, clearly (i) it must value both $\hat{A}_u$ and $\hat{A}_{f_u}$, and (ii) it must envy $\hat{A}_u \cup \hat{A}_{f_u}$. 
     
     Suppose $c(u') \ge 2$. Then from (ii), either $u$ or $f_u$ holds a good from $\struct{u'}$ (say $u$). By~\Cref{clm:c-dist-splgoods}, $u$ is then at distance $c(u) - c(u')$ from $u'$ along allocated edges. From (i) and~\Cref{clm:c-dist-easy}, $f_u$ is at distance at most $c(f_u)$ from $u'$ along allocated edges. Thus, there is $u$-$f_u$ path along allocated edges of length $c(u) + c(f_u) - c(u')$ $\le 2t-3$ (since $c(u) \le t-1$ and $c(u') \ge 2$). But then including the (unallocated) edges $E_{u,f_u}$, there is a cycle of length $2t-2$, which is a contradiction, since we assume the graph has girth at least $2t-1$.

     Finally, suppose $c(u') = 1$. Then note that all the goods that $u'$ values are in $\struct{u'}$. Since both $u$ and $f_u$ possess goods that $u'$ values, by~\Cref{clm:c-dist-splgoods}, they are at distance $c(u)-1$ and $c(f_u)-1$ respectively from $u'$ along allocated edges. Since $c(u) \le t-1$, there is a path of length $2t-3$ from $u$ to $f_u$ along allocated edges. But then including the (unallocated) edges $E_{u,f_u}$, there is a cycle of length $2t-2$, which is again a contradiction. Thus, agent $f_u$ is not envied in the allocation $A$. This completes the proof.
\end{proof}

We now can easily prove that the allocation obtained is EFX.

\begin{theorem}
    \Cref{alg:chromaticEFX} obtains an EFX allocation for agents with cancellable valuations in polynomial time.
    \label{thm:c-main}
\end{theorem}

\begin{proof}
    The proof for the running time for~\Cref{alg:chromaticEFX} is obtained easily, since for cancellable valuations $\cac{u}{S}$ runs in polynomial time, and the other steps are straightforward. 

    We now show that the algorithm obtains an EFX allocation. As before, the proof is by induction on the iterations of the inner for loop. For the base case, the statement clearly holds for the empty allocation. Then consider an agent $u$, and let $\hat{A}$ and $A$ be the allocations before and after $\struct{u}$ is resolved.

    Consider first an envy edge from agent $w$ to agent $u'$ in allocation $\hat{A}$. We will show that this envy remains EFX in allocation $A$ as well. By~\Cref{lem:c-newenvy}, agent $u'$ must be resolved. No agent in $\struct{u}$ is resolved in $\hat{A}$, hence $u' \not \in \struct{u}$. Since only agents in $\struct{u}$ have their allocations modified, agent $u'$'s allocation is not modified. Agent $w$ is either also not in $\struct{u}$ (in which case its allocation also does not change), or is in $\struct{u}$ (in which case its value does not decrease). In either case, either $w$ does not envy $u'$ in $A$, or the envy remains EFX.

    By~\Cref{lem:c-newenvy}, the only envy edge that may be in $A$ but not in $A'$ is from $f_u$ to $u$. This happens in Line~\ref{line:b-condn1-S}, when $u$ gets $S_{u,f_u}$. But then agent $f_u$ gets $\bar{S}_{u,f_u}$. Since agent $f_u$ cut the bundle $E_{u,f_u}$, and agent $u$ has no goods besides $S_{u,f_u}$, the envy is EFX. 
\end{proof}

\section{Conclusion}

Our paper extends the work on EFX existence to important classes of multi-graphs. We view our results as a significant extension over the work by Christodolou et al.~\cite{ChristodolouFKS23}, since their results only allow for a single edge between adjacent vertices. Our work also shows that the well-known cut-and-choose paradigm can be leveraged to obtain results in fairly general settings. 

Clearly, the setting for general multi-graphs remains open. We also find it interesting that the case of multi-graphs with 3 agents does not lend itself to a simpler algorithm than for general 3 agents. Coming up with a simpler algorithm here may be crucial to extending our work to general multi-graphs.

\section*{Acknowledgments}
The authors acknowledge support from the Department of Atomic Energy, Government
of India, under project no. RTI4001

\bibliographystyle{alpha}
\bibliography{References}


\end{document}